\documentclass[preprint]{elsarticle}

\usepackage[a4paper,left=22mm,top=25mm,right=22mm,bottom=35mm, marginparsep=.1in]{geometry}

\usepackage[
 	colorlinks=true,
	urlcolor=black,
	linkcolor=blue
]{hyperref}

\biboptions{sort&compress}

\usepackage{booktabs}

\usepackage{amsmath, amsfonts, amsthm, amssymb}
\usepackage{mathtools}
\usepackage{mathrsfs}
\newtheorem{theorem}{Theorem}
\newtheorem{corollary}[theorem]{Corollary}
\newtheorem{lemma}[theorem]{Lemma}
\newtheorem{proposition}[theorem]{Proposition}
\newdefinition{definition}[theorem]{Definition}

\newcommand{\CS}[1][]{\ensuremath{\mathfrak{C}_{#1}}}
\newcommand{\Tsys}{\ensuremath{\mathfrak{T}}}
\DeclareMathOperator{\CC}{\mathtt{C}}
\newcommand{\overlaps}{\between}
\newcommand{\Hasse}{\ensuremath{\mathcal{H}}}

 \usepackage[dvipsnames]{xcolor}
 \newcommand{\REV}[1]{\begingroup#1\endgroup}

\title{NOC NOC, who's there? Clustering systems of tree-child and normal networks}
\author[1]{Anna Lindeberg}
\ead{anna.lindeberg@math.su.se}
\author[2]{Marc Hellmuth}
\ead{marc.hellmuth@uni-leipzig.de}
\affiliation[1]{organization={Department of Mathematics, Faculty of Science, Stockholm University},
addressline={10691 Stockholm},
country={Sweden}}
\affiliation[2]{organization={Theoretical Computer Science Group, Faculty of Mathematics and Computer Science, Leipzig University},
addressline={Augustusplatz 10, 04109 Leipzig},
country={Germany}}

\begin{document}
\begin{abstract}
	Clustering systems provide a natural way to encode structural information contained in
	phylogenetic networks. In this note, we study the clustering systems of normal and
	tree-child networks through an overlap-based property of set systems, called
	\emph{not-overlap-covered (NOC)}. We show that the NOC property is equivalent to
	inclusion-visibility, a memberwise formulation of the strict-compatibility condition previously
	used for tree-child clustering systems.
	
	We characterize normal networks as precisely the semi-regular networks whose clustering systems
	satisfy NOC. Consequently, a clustering system is realized by a normal network if and only if it
	satisfies NOC, or equivalently, if every one of its clusters is inclusion-visible. In this case,
	the Hasse diagram provides a canonical normal realization. 
	These are exactly the clustering systems realized by tree-child networks.

	The NOC formulation yields a sharp quadratic upper bound on the number of distinct clusters of
	tree-child and normal networks and a direct polynomial-time recognition algorithm.
	\REV{Finally, we explore several consequences of the NOC perspective beyond the phylogenetic setting.
 These include connections to the enumeration of normal networks,
	an order-theoretic interpretation of inclusion-visibility, structural properties of NOC set systems,
	and a tractable special case of Minimum Set Cover, which is NP-hard in general.}
\end{abstract}

\maketitle

\section{Introduction}

Evolutionary histories are often modeled with rooted graphs and were traditionally assumed to be
tree-like. However, when attempting to capture biological mechanisms such as horizontal gene
transfer or hybridization, rooted trees no longer provide a sufficiently general mathematical
framework. This has given rise to the theory of \emph{phylogenetic networks}, which studies directed
acyclic graphs (DAGs) with a unique root, where the set of leaves is identified with a group of
now-living taxa (i.e., species, genes or groups of species). Although motivated by the life
sciences, this area offers a rich variety of interesting questions in discrete and combinatorial
mathematics; see, for example, the books \cite{Dress:2011,Huson:2010} or the surveys
\cite{Huson:2011,Linz:2026} for an introduction.

In this context, the study of clustering systems associated with phylogenetic networks has received
increasing attention in recent years; see, for example,
\cite{Hellmuth:2023,Dai:2026,Barthelemy:08,Brucker:09,Bertrand:14,ALRR:14,Changat:2025,Nakhleh:05}.
To be more precise, a \emph{cluster} $\CC_N(v)$ in a phylogenetic network $N$ is the set of leaves
that are reachable from a given vertex $v$ by a directed path. Biologically, $\CC_N(v)$ can be
interpreted as the set of present-day taxa that may have inherited traits from the ancestral species
represented by $v$. Collecting the clusters associated with all vertices of $N$ yields
$\CS[N]=\{\CC_N(v)\mid v\in V(N)\}$, the \emph{clustering system} of $N$. More generally, a clustering
system on a finite ground set $X$ is a family $\CS$ of subsets of $X$ such that $X\in\CS$,
$\emptyset\notin\CS$, and $\{x\}\in\CS$ for every $x\in X$. Clustering systems provide a natural way
of encoding structural information contained in a phylogenetic network. Understanding which set
systems arise from particular classes of networks is therefore of both theoretical and practical
interest. The classical example is the bijective correspondence between rooted phylogenetic trees
with no vertices of out-degree one and \emph{hierarchies}. A clustering system $\CS$ is a hierarchy
if $ C\cap D\in\{\emptyset,C,D\} $ for all $C,D\in\CS$, and this condition holds precisely for the
clustering systems of rooted phylogenetic trees. For more general classes of phylogenetic networks,
however, the associated clustering systems are substantially more complex and comparatively few
analogous characterizations are known. 

This short note has a single main purpose: to characterize the clustering systems of so-called
\emph{normal networks}. Normal networks \cite{Willson:2010} form a prominent class of phylogenetic
networks; see \cite{Francis:2025} for a recent overview. They are characterized by two properties:
they are shortcut-free, meaning that no arc is redundant in the sense that there exists an
alternative directed path between its endpoints and they satisfy the tree-child property, meaning
that every non-leaf vertex has a child with exactly one parent. Normal networks thus generalize
rooted phylogenetic trees while retaining a number of attractive structural properties, including
reconstructability and statistical identifiability.

The clustering systems of tree-child networks have previously been studied by Alcal{'a} et
al.~\cite{ALRR:14}. They characterized those clustering systems that can be realized by
tree-child reticulate networks in terms of a property called \emph{strict-compatibility} and
provided a polynomial-time recognition and reconstruction algorithm. Their work is formulated for
a more restricted class of reticulate networks than the network framework considered here.

The main result of this paper, Theorem~\ref{thm:normal-characterization}, gives a characterization
of normal networks in this more general setting. To this end, we formulate strict-compatibility
memberwise in terms of what we call \emph{inclusion-visibility} and introduce an equivalent
overlap-based property of arbitrary set systems, called \emph{not-overlap-covered (NOC)}. We show
that a network $N$ is normal if and only if it is semi-regular and every cluster in $\CS[N]$ is
inclusion-visible, or equivalently, if $\CS[N]$ satisfies the NOC property. As a consequence, a
clustering system admits a normal realization if and only if it satisfies NOC. In this case, the Hasse
diagram provides a canonical normal realization.

The NOC formulation also makes additional structural consequences transparent. In particular, we
derive a sharp quadratic upper bound on the number of distinct clusters of tree-child and normal
networks and obtain a direct polynomial-time recognition algorithm. We further record bijective
correspondences with canonical classes of normal networks. Since NOC is defined for arbitrary set
systems, it also leads to questions and consequences beyond the \REV{phylogenetic setting. 
We discuss an enumerative interpretation of the resulting bijections, an order-theoretic
view of inclusion-visibility, hereditary and laminar subclasses of NOC set systems, and an
application to Minimum Set Cover, which is NP-hard in general but becomes polynomial-time solvable
for set systems satisfying the NOC property.}

\section{Preliminaries}

We primarily follow the terminology used in \cite{Hellmuth:2023}.
We consider directed acyclic graphs (DAGs) $G=(V,E)$ with vertex set $V(G)=V$ and arc set $E(G)=E$.
If $(u,v)\in E$ is an arc in $G$, then $u$ is a \emph{parent} of $v$ and $v$ is a \emph{child} of
$u$. A \emph{root} in a DAG is a vertex with no parents. A \emph{leaf} in a DAG is a vertex with no
children. A \emph{network $N$ on $X$} is a DAG with a unique root and leaf-set $X$. For a network
$N$, we write $u\preceq_N v$ if there is a directed path from $v$ to $u$. Note that $u\preceq_N v$
also allows for $u=v$, in which case the directed path just consists of the single vertex $u$. If
$u\preceq_N v$ or $v\preceq_N u$ then $u$ and $v$ are \emph{$\preceq_N$-comparable}. An arc
$(u,v)$ is a \emph{shortcut} if there exists a directed path from $u$ to $v$ that does not contain
the arc $(u,v)$. A network is \emph{shortcut-free} if it does not contain shortcuts. A child
$v$ of $u$ is a \emph{tree-child}, if $v$ has only one parent, namely $u$. A network $N$ on $X$ is
\emph{tree-child} if every vertex $v \in V(N) \setminus X$ has a tree-child \cite{Cardona:2009}, and
$N$ is \emph{normal} if it is tree-child and shortcut-free \cite{Willson:2010}. A vertex $v$ of $N$
is \emph{visible (in $N$)} if there exists a leaf $x\in X$ such that every $\rho x$-path of $N$
passes through $v$, where $\rho$ denotes the unique root of the network $N$. 
By \cite[Lem.~2]{Cardona:2009}, a network is tree-child if and only if all its vertices are visible.
Together with the definition of normal networks, this immediately yields the following.
\begin{theorem}\label{thm:char-normal-vis}
A network is normal if and only if it is shortcut-free and all its vertices are visible.
\end{theorem} 

A \emph{set system $\CS$ on $X$} is a non-empty subset of the powerset of $X$. A set system $\CS$ on
$X$ is a \emph{clustering system} if $X\in \CS$, $\emptyset\notin\CS$, and $\{x\}\in \CS$ for all $x\in
X$. Two sets $A,B\in \CS$ \emph{overlap}, in symbols $A\overlaps B$, if $A\cap B \notin \{\emptyset,
A,B\}$. For a network $N$ on $X$, we define  $\CS[N]=\{\CC_N(v)\mid v\in V(N)\}$
as the set system on $X$ that
contains each \emph{cluster} $\CC_N(v) = \{ x \mid x\in X, x\preceq_N v \}$ of $N$. 
For each network $N$, $\CS[N]$ is a clustering system \cite[Lem~14]{Hellmuth:2023}. A clustering system $\CS$ is \emph{realized}
by a network $N$ if $\CS[N]=\CS$.

An \emph{isomorphism} between two networks $N$ and $N'$ on $X$ is a bijection
$\psi\colon V(N)\to V(N')$ such that $(u,v)$ is an arc of $N$ if and only if 
$(\psi(u),\psi(v))$ is an arc of $N'$, and such that $\psi(x)=x$ for every $x\in X$. 
By definition, network isomorphisms are understood to preserve the leaf labels; that is, an isomorphism
between two networks on $X$ restricts to the identity on $X$. Accordingly, uniqueness of a network
on $X$ is always understood up to such an isomorphism.

The \emph{cover digraph} or \emph{Hasse diagram} $\mathscr{H}(\CS)$
of a clustering system $\CS$ is the
DAG with vertex set $\CS$ and arc set containing an arc $A\to B$ for $A,B\in\CS$ if and only if (i)
$B\subsetneq A$ and (ii) there is no $C\in\CS$ with $B\subsetneq C\subsetneq A$. 
We use $\Hasse(\CS)$ for the network on $X$ obtained from this Hasse diagram $\mathscr{H}(\CS)$ by
relabeling each singleton leaf $\{x\}$ by $x$.  
A network $N$ is \emph{regular} if $ N\simeq\Hasse(\CS[N])$ \cite{Baroni:05}. 
By Proposition~2 of \cite{Hellmuth:2023}, every clustering system $\CS$ is realized 
by a unique regular network, namely by $\Hasse(\CS)$. There are, however, many 
non-regular networks that realize the same clustering system; see 
e.g.\ \cite[Fig.~8]{Hellmuth:2023} or \cite[Fig.~2]{LH:25}.

A network $N$ is \emph{semi-regular} if it is shortcut-free and satisfies the following
\emph{path-comparability-condition (PCC)}: For all $u, v \in V(N)$ it holds that $u$ and $v$ are
$\preceq_N$-comparable if and only if $\CC_N(u) \subseteq \CC_N(v)$ or $\CC_N(v) \subseteq
\CC_N(u)$. We note that Theorem~2 of \cite{Hellmuth:2023} shows that semi-regular networks
generalize the class of regular networks. More precisely, regular networks
are exactly those semi-regular networks that do not contain vertices with a single child.

\section{Inclusion-visibility and the NOC property}

In this section, we introduce two properties of set systems that will be used to characterize
normal networks and their clustering systems. The first is closely related to the notion of
\emph{strict-compatibility} introduced by Alcalá et al.~\cite{ALRR:14} for families of clusters.
We formulate the relevant condition as a property of individual members of an arbitrary set system.

\begin{definition}\label{def:inclusion-visible}
	Let $\CS$ be a set system and $C\in\CS$. Then, $C$ is \emph{inclusion-visible (in
	$\CS$)} if there exists some $x\in C$ such that for all $D\in\CS$ with $x\in D$, we have
	$D\subseteq C$ or $C\subseteq D$. In this case, $x$ is an \emph{inclusion-visible witness} for $C$.
\end{definition}

For clustering systems, inclusion-visibility corresponds exactly to the notion of
\emph{strict-compatibility} introduced in \cite{ALRR:14}. More precisely, Alcal{\'a} et al.\ call a
clustering system $\CS$ \emph{strict-compatible} if, for every $C\in\CS$, there exists some $x\in C$
such that every $D\in\CS$ containing $x$ is compatible with $C$, that is, if $C$ and $D$ are
disjoint or one contains the other. Since $x\in C\cap D$, the sets $C$ and $D$ cannot be disjoint.
Hence, compatibility is equivalent in this situation to $D\subseteq C$ or $C\subseteq D$.
Consequently, $\CS$ is strict-compatible if and only if every $C\in\CS$ is inclusion-visible. Thus,
inclusion-visibility may be viewed as the memberwise version of strict-compatibility, extended here
from clustering systems to arbitrary set systems. The following result, which generalizes
\cite[Lem.~4.1]{ALRR:14}, motivates our terminology ``inclusion-visible'': for arbitrary networks,
visibility of a vertex implies inclusion-visibility of its cluster, while for semi-regular networks
the two properties are equivalent.

\begin{lemma}\label{lem:visible-semiregular}
    Let $N$ be a network and let $v\in V(N)$. If $v$ is visible, then
    $\CC_N(v)$ is inclusion-visible in $\CS[N]$. If $N$ is semi-regular, then $v$ is
    visible if and only if $\CC_N(v)$ is inclusion-visible in $\CS[N]$.
\end{lemma}
\begin{proof}
	Let $N$ be a network with root $\rho$, let $v\in V(N)$ and put $C=\CC_N(v)$. First assume that
	$v$ is visible.
	Then, there exists some $x\in C$ such that every $\rho x$-path in $N$ contains $v$. Let
	$D\in\CS[N]$ be any cluster with $x\in D$. Then, $D=\CC_N(w)$ for some $w\in V(N)$.
	Since $x\in D$, there is a
	directed path from $w$ to $x$, and hence a $\rho x$-path $P'$ containing $w$, obtained by
	concatenating any $\rho w$-path with any $wx$-path.
	Since every $\rho x$-path in $N$ contains $v$, it holds that $v$ is a vertex of $P'$, and
	hence, $v\preceq_N w$ or $w\preceq_N v$. In either case, we have that either $C\subseteq D$ or
	$D\subseteq C$. Hence, $C$ is inclusion-visible in $\CS[N]$.
	
	Now, assume that $N$ is semi-regular. By the first part, it remains to prove the
	\emph{if-direction} of the second statement.
 	Hence, suppose that $\CC_N(v)$ is inclusion-visible in $\CS[N]$. Then, there exists some
	inclusion-visible witness $x$ for $C$. Thus,
	for all $D\in\CS[N]$ with $x\in D$, we have $D\subseteq C$ or $C\subseteq D$.
	Let $P=v_0v_1\ldots v_k$ be an arbitrary $\rho x$-path in $N$, where $v_0=\rho$ and $v_k=x$.
	Since $x\preceq_N v_i$, we have  $x\in \CC_N(v_i)$ for all $i\in\{0,\ldots,k\}$. 
	By the choice of $x$, it follows that $\CC_N(v_i)\subseteq \CC_N(v)$ or
	$\CC_N(v)\subseteq \CC_N(v_i)$ for all $i\in\{0,\ldots,k\}$. Since $N$ satisfies 
	(PCC), the vertices $v$ and $v_i$ are therefore $\preceq_N$-comparable for all
	$i\in\{0,\ldots,k\}$. The latter together with the fact that $x=v_k\preceq_N v\preceq_N v_0=\rho$
	implies that there is some $j\in\{0,\ldots,k-1\}$ such that $v_{j+1}\preceq_N v\preceq_N v_j$.
	Since $v_{j+1}$ and $v_j$ are adjacent in $N$ and $N$ has no shortcuts, we conclude that
	$v=v_{j+1}$ or $v=v_j$. In either case, $v$ is a vertex of $P$. Since $P$ was an
	arbitrary $\rho x$-path in $N$, we conclude that $v$ is a visible vertex.
\end{proof}

We now introduce an overlap-based property of set systems, which will provide an equivalent
formulation of inclusion-visibility.

\begin{definition}
	Let $\CS$ be a set system. Then, $C\in\CS$  is \emph{not-overlap-covered (NOC)} if
		\[C\not\subseteq\bigcup_{\substack{D\in\CS:\, D\overlaps C}} D.\]
	An element $x\in C\setminus\bigcup_{D\in\CS,D\overlaps C} D$ is a \emph{NOC witness} for $C$. 
	We say that $\CS$ has the \emph{NOC property} if every $C\in\CS$ is NOC.
\end{definition}

By definition, a NOC witness is an element of $C$ that does not occur
in any set overlapping $C$ and thus certifies that $C$ is NOC.
The existence of such witnesses imposes a strong restriction on the size of set systems
with the NOC property.

\begin{lemma}\label{lem:NOC-size}
	Let $\CS$ be a set system on an $n$-element ground set $X$. If $\CS$ has the NOC property, then
	\[|\CS|\leq \frac{n(n+1)}{2}\] and this bound is sharp.
\end{lemma}
\begin{proof}
Let $\CS$ be a set system on an $n$-element ground set $X$
with the NOC property. For each $k\in\{1,\ldots,n\}$,
define $\CS[k]=\{S\in\CS\colon |S|=k\}$. We first show that $ |\CS[k]|\leq n-k+1$. To this end,
suppose that $\CS[k]=\{S_1,\ldots,S_m\}$. Since $\CS$ has the NOC property, each $S_i$ has a
NOC witness; choose one and denote it by $x_i$. We claim that $x_i$ is contained in no other member
of $\CS[k]$. Indeed, suppose that $ x_i\in S_j$ for some $j\neq i$. Then $S_i\cap
S_j\neq\emptyset$. Since $S_i$ and $S_j$ are distinct sets of the same cardinality, neither can
contain the other. Hence $S_i\overlaps S_j$; contradicting the fact that $x_i$ is a NOC witness for
$S_i$. Thus, every member $S_i$ of $\CS[k]$ has an element $x_i$ that occurs in no other member of
$\CS[k]$. In particular, the witnesses $x_1,\ldots,x_m$ are pairwise distinct.

Now fix one set, say $S_1$. It consists of its witness $x_1$ together with $k-1$ further elements.
None of these $k-1$ elements can be one of the witnesses $x_2,\ldots,x_m$, since $x_i$ occurs only
in $S_i$ among the members of $\CS[k]$. Consequently, the ground set $X$ contains at least the $m$
distinct witnesses $x_1,\ldots,x_m$ together with the $k-1$ remaining elements of $S_1$. Hence
$n\geq m+k-1$ and therefore $ |\CS[k]|=m\leq n-k+1. $
Summing over all possible cardinalities gives $ |\CS| = \sum_{k=1}^n |\CS[k]| \leq
\sum_{k=1}^n(n-k+1) = n+(n-1)+\cdots+1 =\frac{n(n+1)}{2} $. 

To see that the bound is sharp, consider
$\CS=\bigl\{\{1,\ldots,k-1\}\cup\{j\}:1\leq k\leq j\leq n\bigr\}$.
For each such set, $j$ is a NOC witness, and
$|\CS|=\sum_{k=1}^n(n-k+1)=\frac{n(n+1)}{2}$.
\end{proof}

As we shall see, inclusion-visibility and the NOC property are equivalent.
Together with the observation above, this also provides an alternative characterization of strict-compatibility for
clustering systems.

\begin{lemma}\label{lem:clusters-visible-noc}
	Let $\CS$ be a set system and let $C\in\CS$. Then, $C$ is inclusion-visible if and only if
	$C$ is NOC. In particular, $x$ is an inclusion-visible witness for $C$ if and only if $x$
	is a NOC witness for $C$.
\end{lemma}
\begin{proof}
	Let $\CS$ be a set system and $C\in\CS$. First assume that $C$ is inclusion-visible. Let
	$x\in C$ be such that for all $D\in\CS$ with $x\in D$, we have $D\subseteq C$ or $C\subseteq D$.
	By contraposition, if $D\in\CS$ satisfies $D\overlaps C$, then $x\notin D$. Consequently, $C$ is NOC.

	Conversely, assume $C$ is NOC. Then, there exists some element 
	$x\in C\setminus (\cup_{D\overlaps C} D)$. Thus, $x\notin D$ for all
	$D\in\CS$ with $D\overlaps C$. Consequently, every $D\in\CS$ containing $x$ is comparable with
	$C$ under inclusion, that is, $D\subseteq C$ or $C\subseteq D$. Hence, $C$ is inclusion-visible.
	
	The arguments used here, in particular, imply the last statement of the lemma.
\end{proof}

\section{Characterizations and algorithmic consequences}

We first establish a new characterization of normal networks.

\begin{theorem}\label{thm:normal-characterization}
Let $N$ be a network. Then, the following statements are equivalent:
\begin{enumerate}[(1)]
\item $N$ is normal.
\item $N$ is semi-regular and every cluster in $\CS[N]$ is inclusion-visible.
\item $N$ is semi-regular and $\CS[N]$ has the NOC property.
\end{enumerate}
\end{theorem}
\begin{proof}
Let $N$ be a network. First assume that $N$ is normal. By Corollary~14 of \cite{Hellmuth:2023},
$N$ is semi-regular. Moreover, Theorem~\ref{thm:char-normal-vis} implies that every vertex in
$N$ is visible. Hence, by Lemma~\ref{lem:visible-semiregular}, every
cluster in $\CS[N]$ is inclusion-visible. Thus, (1) implies (2).
Conversely, assume that (2) holds. By Lemma~\ref{lem:visible-semiregular},
every vertex in $N$ is visible. Since $N$ is semi-regular, it is
also shortcut-free. Therefore, Theorem~\ref{thm:char-normal-vis} implies that $N$ is normal.
Hence, (2) implies (1). Finally, Statements (2) and (3) are equivalent by
Lemma~\ref{lem:clusters-visible-noc}.
\end{proof}

A characterization of tree-child realizability in terms of strict-compatibility was established in
\cite[Thm.~4.6]{ALRR:14}. In view of the equivalence between strict-compatibility and
inclusion-visibility discussed above, this already characterizes the corresponding clustering systems
in terms of inclusion-visibility. Combining this with the NOC formulation and
Theorem~\ref{thm:normal-characterization} yields the following characterization in our setting.

\begin{corollary}\label{cor:NOC-char}
    Let $\CS$ be a clustering system. Then, the following statements are equivalent: 
    \begin{enumerate}[(1)]
        \item Every cluster in $\CS$ is inclusion-visible.
        \item $\CS$ has the NOC property.
        \item There exists a normal network $N$ such that $\CS[N]=\CS$.
        \item There exists a tree-child network $N$ such that $\CS[N]=\CS$.
    \end{enumerate}
    In particular, if (1) or (2) holds, then 
   $N\coloneqq\Hasse(\CS)$ is a normal network with $\CS[N]=\CS$.
\end{corollary}
\begin{proof}
  Let $\CS$ be a clustering system. The equivalence of Statements (1) and (2) is given by
  Lemma~\ref{lem:clusters-visible-noc}. Assume that (1) holds. By Proposition~2 in 
  \cite{Hellmuth:2023}, there exists a unique regular network $N$ such that $\CS[N]=\CS$. 
  In particular, as shown in the proof of Proposition~2 in \cite{Hellmuth:2023}, 
  this network is given by the Hasse diagram $N\coloneqq\Hasse(\CS)$. 
  Since every cluster in $\CS$ is inclusion-visible, Lemma~\ref{lem:visible-semiregular} 
  implies that every vertex in $N$ is visible. Since $N$  is regular, it is in particular 
  semi-regular (cf.\ \cite[Thm.~2]{Hellmuth:2023}). The latter two conditions together with 
  Theorem~\ref{thm:normal-characterization} imply that $N$ is normal. Hence, (1) implies (3).
 
  Since every normal network is tree-child, Statement (3) immediately implies (4).
  Finally, assume that (4) holds, and let $N$ be a tree-child network with $\CS[N]=\CS$. As observed
  in \cite[Lem.~2]{Cardona:2009}, every vertex of a tree-child network is visible.  Hence, by
  Lemma~\ref{lem:visible-semiregular}, every cluster in $\CS[N]=\CS$ is inclusion-visible. Thus, (4)
  implies (1).
\end{proof}

The existence of a normal realization is also implicit in \cite{ALRR:14}: for a
strict-compatible clustering system, the associated cluster network is tree-child
\cite[Thm.~4.6]{ALRR:14} and shortcut-free \cite[Rem.~3.5]{ALRR:14} (therein called ``non-redundant''), 
hence normal in our terminology. The new aspects here are the NOC formulation for arbitrary set systems and its
connection with the canonical regular realization $\Hasse(\CS)$.
As a further consequence, we obtain the following sharp
quadratic bound on the number of distinct clusters of such networks.

\begin{corollary}
	Let $N$ be a tree-child or normal network on an $n$-element leaf set $X$. 
	Then $|\CS[N]|\leq \frac{n(n+1)}{2}$. 
	If, in addition, $N$ is regular, then $|V(N)|\leq \frac{n(n+1)}{2}$.
	In both cases, the bound is sharp.
\end{corollary}
\begin{proof}
By Corollary~\ref{cor:NOC-char}, $\CS[N]$ has the NOC property. Hence,
Lemma~\ref{lem:NOC-size} gives
$|\CS[N]|\leq \frac{n(n+1)}2$.
If $N$ is regular, we have $N\simeq \Hasse(\CS[N])$ and it follows that the 
vertices of $N$ are in bijection with the clusters in $\CS[N]$.
Therefore $|V(N)|=|\CS[N]|$ and the second assertion follows.

In both cases, sharpness follows from Lemma~\ref{lem:NOC-size} together with
Corollary~\ref{cor:NOC-char}, since the extremal set system constructed in the proof of
Lemma~\ref{lem:NOC-size} is a clustering system and is realized by the regular normal network
$\Hasse(\CS)$.
\end{proof}

Polynomial-time recognition of clustering systems realizable by tree-child networks was already
obtained in \cite{ALRR:14} via the construction of the associated cluster network. The NOC
formulation yields a direct recognition algorithm that avoids constructing a network and gives the
following slightly improved running-time bound.

\begin{corollary}\label{cor:noc-polytime}
Let $\CS$ be a clustering system on an $n$-element ground set $X$. Then it can be decided in 
$ O(\min\{|\CS|^2n,\; n^5\}) $
time whether $\CS$ has the NOC property. Consequently, it can be decided in polynomial time
whether there exists a normal network $N$ such that $\CS[N]=\CS$. If such a network exists, a normal
network realizing $\CS$ can be constructed in polynomial time.
\end{corollary}
\begin{proof}
Let $m=|\CS|$. We assume that the elements of $X$ are ordered and that every set $C\in\CS$ is
represented by its characteristic vector in $\{0,1\}^n$. For each $C\in\CS$, compute
$U_C=\bigcup_{{D\in\CS,D\overlaps C}} D$. By definition, $C$ is NOC if and only if $C\setminus
U_C\neq\emptyset$. Thus, $\CS$ has the NOC property if and only if this condition holds for every
$C\in\CS$. For each ordered pair $(C,D)\in\CS\times\CS$, whether $C$ and $D$ overlap can be decided
in $O(n)$ time by a single scan through their characteristic vectors. 
At the same time, the sets $D$ overlapping $C$ can be
added to the union $U_C$. Hence, all sets $U_C$ and the corresponding tests $C\setminus
U_C\neq\emptyset$ can be computed in $O(m^2n)$ time. On the other hand, by
Lemma~\ref{lem:NOC-size}, every NOC set system on an $n$-element ground set contains at most $
\frac{n(n+1)}2 $ members. Hence, if $m>\frac{n(n+1)}2$, we may immediately conclude that $\CS$ is
not NOC. Otherwise, $m=O(n^2)$, and the above algorithm runs in $O(n^5)$ time. Therefore, the NOC
property can be recognized in $ O(\min\{m^2n,n^5\}) = O(\min\{|\CS|^2n,n^5\}) $ time.

By Corollary~\ref{cor:NOC-char}, $\CS$ is the clustering system of a normal network
if and only if $\CS$ has the NOC property. Moreover, whenever this is the case, the
Hasse diagram $\Hasse(\CS)$ is a normal network realizing $\CS$. Since $\Hasse(\CS)$ is obtained by
testing, for each pair $A,B\in\CS$, whether $B\subsetneq A$ and whether there exists some $C\in\CS$
with $B\subsetneq C\subsetneq A$, it can be constructed in polynomial time.
\end{proof}

We conclude this section by recording two bijective correspondences that follow from the
NOC characterization together with known uniqueness results for regular and separated
phylogenetic networks. These correspondences are closely related to the bijection between
cluster networks and regular networks established in \cite[Thm.~3.8]{ALRR:14}. To state them, we
recall three standard classes of networks. A network is \emph{phylogenetic} if no vertex has
exactly one child and at most one parent, and it is \emph{strong-phylogenetic} if no vertex has
exactly one child. Moreover, a network is \emph{separated} if every vertex with more than one parent has exactly one child.

\begin{corollary}\label{cor:bijection-normal-noc}
Let $X$ be a finite set, and let
\[
\mathsf{NOC}(X) =
\{\CS \mid \CS \text{ is a clustering system on } X \text{ satisfying the NOC property}\}.
\]
Then, the map $ N\mapsto \CS[N] $
induces a bijection from each of the following two classes, considered up to isomorphism, onto $\mathsf{NOC}(X)$:
\begin{enumerate}[(1)]
\item Strong-phylogenetic normal networks on $X$.
\item Phylogenetic  and separated normal networks on $X$.
\end{enumerate}
Consequently, there is a canonical bijection between the classes in (1) and (2): a network $N$ in
(1) is mapped to the unique network $N'$ in (2) satisfying $ \CS[N']=\CS[N]$.
\end{corollary}
\begin{proof}
By Theorem~\ref{thm:normal-characterization}, the clustering system of every normal network
satisfies the NOC property. Hence, for each of the classes in (1) and (2), the map $ N\mapsto
\CS[N] $ indeed takes values in $\mathsf{NOC}(X)$.

We first consider the class in (1). Let $\CS\in\mathsf{NOC}(X)$. By
Corollary~\ref{cor:NOC-char}, there is a normal network realizing $\CS$, namely
$N\coloneqq\Hasse(\CS)$. Clearly, $N$ is regular. For normal networks,
regularity is equivalent to being strong-phylogenetic \cite[Thm.~3.14]{HLM:26}, and therefore $N$ is
a strong-phylogenetic normal network. Thus, the map $N\mapsto\CS[N]$ is surjective onto
$\mathsf{NOC}(X)$. Conversely, suppose that $N_1$ and $N_2$ are strong-phylogenetic normal networks
on $X$ with $\CS[N_1]=\CS[N_2]=\CS$. Since both $N_1$ and $N_2$ are normal and strong-phylogenetic,
Theorem~3.14 of \cite{HLM:26} implies that both networks are regular.
Proposition~2 of \cite{Hellmuth:2023} states that a clustering system admits, up to isomorphism, a
unique regular realization. Since both $N_1$ and $N_2$ are regular realizations of the same
clustering system $\CS$, they are isomorphic. Hence, the map $N\mapsto\CS[N]$ is injective on the
class in (1).

We now consider the class in (2). By Theorem~6 of \cite{Hellmuth:2023}, every clustering system
$\CS$ on $X$ is realized by a unique semi-regular, phylogenetic, and separated network. Let
$\CS\in\mathsf{NOC}(X)$, and let $N$ denote this unique realization. Since $N$ is semi-regular and
$\CS[N]=\CS$ has the NOC property, Theorem~\ref{thm:normal-characterization} implies that $N$ is
normal. Consequently, $N$ is a phylogenetic and separated normal network, that is, a network in the
class in (2). Hence, the map $N\mapsto\CS[N]$ is surjective onto $\mathsf{NOC}(X)$. Conversely,
suppose that $N_1$ and $N_2$ are phylogenetic and separated normal networks on $X$ with $
\CS[N_1]=\CS[N_2]=\CS$. Proposition~8 of \cite{Hellmuth:2023} states that if a clustering system
admits a realization by a phylogenetic and separated normal network, then that normal network is
unique up to isomorphism. In other words, $N_1$ and $N_2$ are necessarily isomorphic. Hence, the map
$N\mapsto\CS[N]$ is injective on the class in (2).

Thus, for each of the classes in (1) and (2), the map $N\mapsto\CS[N]$ induces a bijection onto
$\mathsf{NOC}(X)$.
\end{proof}

\section{Concluding Remarks and Outlook}
We have characterized normal networks in terms of inclusion-visibility and the NOC property,
and used this viewpoint to revisit the clustering systems realized by normal and tree-child
networks. For clustering systems, inclusion-visibility coincides with the strict-compatibility
condition of \cite{ALRR:14}, while the equivalent NOC formulation provides an overlap-based
perspective that applies to arbitrary set systems. This perspective yields, among other
consequences, the sharp quadratic bound and the direct recognition algorithm established above.

The bijections in Corollary~\ref{cor:bijection-normal-noc} also provide a new perspective on the
enumeration of normal networks.  To be more precise, let $q_n$ denote the number of clustering
systems on an $n$-element set $X$ satisfying the NOC property. Then
Corollary~\ref{cor:bijection-normal-noc} implies that $q_n$ is precisely both the number of
strong-phylogenetic normal networks on $X$ and the number of phylogenetic and separated normal
networks on $X$, in each case up to isomorphism.
This correspondence can be refined by taking the number of reticulation vertices into account.
Let $q_{n,k}$ denote the number of (NOC) clustering systems $\CS$ on an $n$-element set $X$ for
which $\Hasse(\CS)$ has exactly $k$ reticulation vertices, that is, vertices with at least two
parents. Since $\Hasse(\CS)$ is the unique regular realization of $\CS$, the number $q_{n,k}$ is
precisely the number of strong-phylogenetic normal networks on $X$ with exactly $k$ reticulation
vertices. The enumeration of phylogenetic networks, in particular tree-child and normal networks, 
has received considerable attention in recent years, both concerning exact enumeration and asymptotic
behavior; see, for example,
\cite{McDiarmid:2015,Fuchs:2019,Fuchs:2021A,Fuchs:2021B,Pons:2021}.
It would therefore be interesting to investigate whether the set-system characterization developed
here can be used to obtain recurrences or explicit formulas for $q_n$ or $q_{n,k}$, and to understand the connections to the recently resolved Pons-Batle identity \cite{Pons:2021,Liu:2026,Yu:2026}.

Inclusion-visibility also admits a natural order-theoretic interpretation. Indeed, a cluster
$C\in\CS$ is inclusion-visible if and only if there exists some $x\in C$ such that $C$ is
comparable, with respect to inclusion, with every member of $\CS|_x=\{D\in\CS\mid x\in D\}$. This
suggests a corresponding notion for an arbitrary poset $(P,\leq)$: an element $y\in P$ may be called
\emph{visible} if there exists a minimal element $x\leq y$ such that $y$ is comparable with every
element of the up-set $x^\uparrow=\{z\in P\mid x\leq z\}$. It may be interesting to investigate
whether this notion yields useful structural consequences for particular classes of posets or
lattices.

For arbitrary set systems, the NOC property has further structural consequences. In
particular, it is hereditary: if a set system $\CS$ has the NOC property,
then every subfamily of $\CS$ has the NOC property. Hence,
as a consequence of Corollary~\ref{cor:NOC-char},
if $\CS$ is the clustering system of a normal, equivalently, of a tree-child network,
then every clustering system $\CS'\subseteq\CS$ is the clustering system of a normal network.

It is also noteworthy that the NOC property generalizes laminar set systems, that is, set systems
in which no two sets overlap. In the setting of clustering systems, laminar set systems are precisely
hierarchies. Indeed, every laminar set system has the NOC property, since for each $C\in\CS$ there
is no set $D\in\CS$ overlapping $C$. The converse does not hold, since NOC set systems may contain
overlapping sets.

Beyond this structural observation, the NOC property also has algorithmic consequences. For
example, consider the classical \emph{Minimum Set Cover} problem. Given a finite set system $\CS$
on a ground set $X$ with $\bigcup_{S\in\CS}S=X$, the task is to find a subfamily
$\Tsys\subseteq\CS$ such that $\bigcup_{T\in\Tsys}T=X$. Such a family is called a \emph{set cover
of $X$}. If $\Tsys$ has minimum cardinality among all set covers of $X$, then it is called a
\emph{minimum set cover of $X$}. Minimum Set Cover is NP-hard in general \cite{Karp:1972}.
In contrast, for set systems satisfying the NOC property, an optimal solution is obtained directly
from the inclusion-maximal members of the set system, as shown in the following result.

\begin{proposition}\label{prop:set-cover}
	Let $\CS$ be a set system on $X$ such that $\cup_{S\in\CS}S=X$ and let $\Tsys\subseteq\CS$.
   	If $\Tsys$ is an inclusion-minimal set cover of $X$, then $\Tsys$ has the NOC property.
	Consequently, all minimum set covers have the NOC property. 
 
   Moreover, if $\Tsys$ comprises the inclusion-maximal elements of $\CS$ and $\CS$ has the NOC 
   property, then $\Tsys$ is a minimum set cover of $X$.
\end{proposition}
\begin{proof}
	Let $\CS$ be a set system on $X$ such that $\cup_{S\in\CS}S=X$ and let
	$\Tsys\subseteq\CS$. Suppose that $\Tsys$ is an inclusion-minimal set cover of $X$. Let
	$T\in\Tsys$. Since $\Tsys$ is inclusion-minimal as a set cover, the family $\Tsys\setminus \{T\}$
	does not cover $X$. Hence, there exists some element $x_T\in T$ such that $x_T\notin S$ for all
	$S\in\Tsys\setminus\{T\}$. In particular, $ x_T\notin S $ for every $S\in\Tsys$ with $S\overlaps
	T$. Therefore, $ x_T\in T\setminus \bigcup_{S\in\Tsys,S\overlaps T}S, $ and hence
	$T$ is NOC. Since $T$ was arbitrary, $\Tsys$ has the NOC property.

    Suppose now that $\CS$ has the NOC 
    property and that $\Tsys$ comprises the inclusion-maximal elements of $\CS$. 
    First observe that $\Tsys$ is a set cover of $X$: every element of every member of
    $\CS$ is contained in some inclusion-maximal member of $\CS$. 
    Now consider any minimum set cover $\Tsys'\subseteq\CS$ of $X$. We show that $|\Tsys|\leq|\Tsys'|$. 
	To that end, note that each $T\in\Tsys$ is NOC, so there exists some
    $x\in T$ such that, for every $S\in\CS$ with $S\overlaps T$, we have $x\notin S$. Since
    $\Tsys'$ is a set cover, there is some element of $\Tsys'$ containing $x$; choose one such set
    and denote it by $T_x$. 
    Since $T_x\in\Tsys'\subseteq\CS$ and $x\in T_x$, the choice of $x$ implies that $T_x$ cannot overlap $T$.
    This together with
    $x\in T\cap T_x$ and inclusion-maximality of $T$ implies that $T_x\subseteq T$. Moreover, any
    $S\in\Tsys$ with $S\neq T$ either overlaps $T$ or is disjoint from $T$. In
    either case, $x\notin S$. The latter and $x\in T_x$ implies that $T$ is the only element of
    $\Tsys$ that contains $T_x$. Defining $\varphi(T)=T_x$ for each $T\in \Tsys$ therefore yields an
    injective map from $\Tsys$ to $\Tsys'$. Hence, $|\Tsys|\leq|\Tsys'|$. Since $\Tsys$ itself is a set cover, it follows
    that $\Tsys$ is a minimum-cardinality set cover of $X$.
\end{proof}

This simple connection with Minimum Set Cover suggests that the NOC property may be of independent
interest beyond its role in phylogenetic network theory. It would be interesting to investigate
which further combinatorial optimization problems become tractable on NOC set systems, and whether
the structural restrictions imposed by NOC lead to useful decomposition or recognition results in
other settings.

\bibliographystyle{elsarticle-num}
\bibliography{references}

@article{ALRR:14,
  author =        {Alcal{\`a}, Adri{\`a} and Llabr{\'e}s, Merc{\`e}
                   and Rossell{\'o}, Francesc and Rullan, Pau},
  title =         {Tree-Child Cluster Networks},
  year =          {2014},
  volume =        {134},
  number =        {1-2},
  journal =       {Fundam. Inf.},
  pages =         {1-15},
  doi =           {10.3233/FI-2014-1087},
}

@Article{Baroni:05,
  author =       {Baroni, Mihaela and Semple, Charles and Steel, Mike},
  title =        {A Framework for Representing Reticulate Evolution},
  journal =      {Annals of Combinatorics},
  volume =       {8},
  pages =        {391-408},
  year =         {2005},
  doi =          {10.1007/s00026-004-0228-0},
}

@Article{Barthelemy:08,
  journal =      {Discr. Appl. Math.},
  volume =       {156},
  year =         {2008},
  pages =        {1237-1250},
  title =        {Binary clustering},
  author =       {Barth{\'e}lemy, Jean-Pierre and Brucker, Fran{\c{c}}ois},
  doi =          {10.1016/j.dam.2007.05.024},
}

@Article{Brucker:09,
  author =       {Brucker, Fran{\c{c}}ois and G{\'e}ly, Alain},
  title =        {Parsimonious cluster systems},
  journal =      {Adv Data Anal Classif},
  year =         {2009},
  volume =       {3},
  pages =        {189-204},
  doi =          {10.1007/s11634-009-0046-7},
}

@InCollection{Bertrand:14,
  author =         {Bertrand, Patrice and Diatta, Jean},
  title =          {Weak Hierarchies: A Central Clustering Structure},
  booktitle =      {Clusters, Orders, and Trees: Methods and Applications},
  editor =         {Aleskerov, Fuad and Goldengorin, Boris and Pardalos,
                    Panos M.},
  publisher =      {Springer},
  address =        {New York},
  year =           {2014},
  pages =          {211-230},
  doi =            {10.1007/978-1-4939-0742-7_14},
}

@article{Cardona:2009,
  author  = {Cardona, Gabriel and Rossell{\'o}, Francesc and
             Valiente, Gabriel},
  doi     = {10.1109/TCBB.2007.70270},
  journal = {IEEE/ACM Trans. Comp. Biol. Bioinf.},
  pages   = {552-569},
  title   = {Comparison of tree-child phylogenetic networks},
  volume  = {6},
  year    = {2009}
}

@article{Changat:2025,
  author  = {Manoj Changat and Ameera Vaheeda Shanavas and Peter F. Stadler},
  title   = {Transit functions and clustering systems},
  journal = {The Art of Discrete and Applied Mathematics},
  volume  = {8},
  number  = {2},
  year    = {2025},
  doi     = {10.26493/2590-9770.1782.ec6}
}

@book{Dress:2011,
  author    = {Dress, Andreas and Huber, Katharina T. and Koolen, Jacobus and Moulton, Vincent and Spillner, Andreas},
  publisher = {Cambridge University Press},
  title     = {Basic Phylogenetic Combinatorics},
  year      = {2011}
}

@misc{Dai:2026,
      title={Minimum Network Level Forced by Hardwired Cluster Data}, 
      author={Shilong Dai and Yangjing Long},
      year={2026},
      eprint={2605.21945},
      archivePrefix={arXiv},
      primaryClass={q-bio.MN},
      url={https://arxiv.org/abs/2605.21945}, 
}

@article{Fuchs:2021A,
  title   = {Counting phylogenetic networks with few reticulation vertices: exact enumeration and corrections},
  author  = {Fuchs, Michael and Gittenberger, Bernhard and Mansouri, Marefatollah},
  journal = {Australasian Journal of Combinatorics},
  volume  = {81},
  number  = {2},
  pages   = {257--282},
  year    = {2021},
  issn    = {2202-3518},
  note    = {CC BY, https://creativecommons.org/licenses/by/4.0/},
  url     = {https://ajc.maths.uq.edu.au/v81.p257},
}

@article{Fuchs:2019,
  title   = {Counting phylogenetic networks with few reticulation vertices: tree-child and normal networks},
  author  = {Fuchs, Michael and Gittenberger, Bernhard and Mansouri, Marefatollah},
  journal = {Australasian Journal of Combinatorics},
  volume  = {73},
  number  = {2},
  pages   = {385--423},
  year    = {2019},
  issn    = {2202-3518},
  note    = {CC BY, https://creativecommons.org/licenses/by/4.0/},
  url     = {https://ajc.maths.uq.edu.au/v73.p385},
}

@article{Fuchs:2021B,
title = {On the asymptotic growth of the number of tree-child networks},
journal = {European Journal of Combinatorics},
author = {Michael Fuchs and Guan-Ru Yu and Louxin Zhang},
volume = {93},
pages = {103278},
year = {2021},
issn = {0195-6698},
doi = {10.1016/j.ejc.2020.103278},
}

@article{Francis:2025,
  author  = {Andrew Francis},
  doi     = {10.1016/j.jtbi.2025.112236},
  journal = {Journal of Theoretical Biology},
  pages   = {112236},
  title   = {"{N}ormal" phylogenetic networks may be emerging as the leading class},
  volume  = {614},
  year    = {2025}
}

@article{Hellmuth:2023,
  author  = {Hellmuth, Marc and Schaller, David and Stadler, Peter F.},
  doi     = {10.1007/s12064-023-00398-w},
  journal = {Theory in Biosciences},
  number  = {4},
  pages   = {301-358},
  title   = {Clustering systems of phylogenetic networks},
  volume  = {142},
  year    = {2023}
}

@misc{HLM:26,
  archiveprefix = {arXiv},
  author        = {Marc Hellmuth and Anna Lindeberg and Vincent Moulton},
  eprint        = {2605.21725},
  primaryclass  = {q-bio.PE},
  title         = {Regularizing and Normalizing {DAG}s and Phylogenetic Networks},
  url           = {https://arxiv.org/abs/2605.21725},
  year          = {2026}
}

@Article{LH:25,
author={Lindeberg, Anna and Hellmuth, Marc},
title={Simplifying and Characterizing {DAG}s and Phylogenetic Networks via Least Common Ancestor Constraints},
journal={Bulletin of Mathematical Biology},
year={2025},
month={Feb},
volume={87},
number={3},
pages={44},
doi={10.1007/s11538-025-01419-z},
}

@book{Huson:2010,
  author    = {Huson, Daniel H and Rupp, Regula and Scornavacca, Celine},
  publisher = {Cambridge University Press},
  title     = {Phylogenetic Networks},
  year      = {2010}
}

@article{Huson:2011,
  author  = {Huson, Daniel H. and Scornavacca, Celine},
  doi     = {10.1093/gbe/evq077},
  issn    = {1759-6653},
  journal = {Genome Biology and Evolution},
  month   = {01},
  pages   = {23-35},
  title   = {A Survey of Combinatorial Methods for Phylogenetic Networks},
  volume  = {3},
  year    = {2011}
}

@inproceedings{Karp:1972,
author = {Karp, Richard M.},
editor = {Miller, Raymond E. and Thatcher, James W. and Bohlinger, Jean D.},
title = {Reducibility among Combinatorial Problems},
bookTitle = {Complexity of Computer Computations: Proceedings of a symposium on the Complexity of Computer Computations},
year = {1972},
publisher = {Springer US},
pages = {85--103},
doi = {10.1007/978-1-4684-2001-2_9},
}

@article{Linz:2026,
  author = {Simone Linz and Kristina Wicke.},
  journal = {Notices of the American Mathematical Society},
  pages = {140--143},
  title = {What is phylogenetics?},
  volume = {73},
  year = {2026},
  doi = {10.1090/noti3299}
}

@InProceedings{Liu:2026,
  author =	{Liu, Hexuan and Wallner, Michael and Yu, Guan-Ru},
  title =	{{A Combinatorial Framework for the Pons-Batle Identity: Young Tableaux, Lattice Paths, and Limit Laws}},
  booktitle =	{37th International Conference on Probabilistic, Combinatorial and Asymptotic Methods for the Analysis of Algorithms (AofA 2026)},
  pages =	{13:1--13:20},
  year =	{2026},
  volume =	{381},
  editor =	{Panagiotou, Konstantinos},
  publisher =	{Schloss Dagstuhl -- Leibniz-Zentrum f{\"u}r Informatik},
  doi =		{10.4230/LIPIcs.AofA.2026.13},
  }

@article{McDiarmid:2015,
	author = {McDiarmid, Colin and Semple, Charles and Welsh, Dominic},
	journal = {Annals of Combinatorics},
	number = {1},
	pages = {205--224},
	title = {Counting Phylogenetic Networks},
	volume = {19},
	year = {2015},
  doi = {10.1007/s00026-015-0260-2}
  }

@InProceedings{Nakhleh:05,
  title =        {Phylogenetic Networks: Properties and Relationship
                  to Trees and Clusters},
  booktitle =    {Transactions on Computational Systems Biology {II}},
  author =       {Nakhleh, Luay  and Wang, Li-San},
  year =         {2005},
  editor =       {Priami, C. and Zelikovsky, A.},
  series =       {Lect. Notes Comp. Sci.},
  volume =       {3680},
  publisher =    {Springer},
  address =      {Berlin, Heidelberg},
  doi =          {10.1007/11567752_6},
  pages =        {82-99},
}

@article{Pons:2021,
	author = {Pons, Miquel and Batle, Josep},
	doi = {10.1038/s41598-021-01166-w},
	journal = {Scientific Reports},
	number = {1},
	title = {Combinatorial characterization of a certain class of words and a conjectured connection with general subclasses of phylogenetic tree-child networks},
	volume = {11},
	year = {2021},
}

@article{Willson:2010,
  author  = {Willson, Stephen J.},
  doi     = {10.1007/s11538-009-9449-z},
  journal = {Bulletin of Mathematical Biology},
  number  = {2},
  pages   = {340-358},
  title   = {Properties of Normal Phylogenetic Networks},
  volume  = {72},
  year    = {2010}
}

@misc{Yu:2026,
      title={A Short Combinatorial Proof of the Pons-Batle Identity for Counting Tree-Child Networks}, 
      author={Hao Yu and Louxin Zhang},
      year={2026},
      eprint={2609.04979},
      archivePrefix={arXiv},
      primaryClass={math.CO},
      url={https://arxiv.org/abs/2609.04979}, 
}

\end{document}